\documentclass[sigconf]{acmart}

\copyrightyear{2021}
\acmYear{2021}
\setcopyright{acmcopyright}
\acmConference[PODS '21] {Proceedings of the 40th ACM SIGMOD-SIGACT-SIGAI Symposium on Principles of Database Systems}{June 20--25, 2021}{Virtual Event, China}
\acmBooktitle{Proceedings of the 40th ACM SIGMOD-SIGACT-SIGAI Symposium on Principles of Database Systems (PODS '21), June 20--25, 2021, Virtual Event, China}
\acmPrice{15.00}
\acmISBN{978-1-4503-8381-3/21/06}
\acmDOI{10.1145/3452021.3458333}

\usepackage{amsmath,amsthm, color, tikz}
\usepackage{url}
\usepackage{tcolorbox}
\usepackage{thmtools} 
\usepackage{bbm}
\usepackage{dsfont}

\usepackage{url}
\usetikzlibrary{positioning}
\usetikzlibrary{shapes}
\usetikzlibrary{decorations.pathreplacing}
\usepackage{enumerate, enumitem} 
\usepackage{hyperref} 
\usepackage{algorithm}
\usepackage{algpseudocode}

\usepackage{mathtools}

\DeclarePairedDelimiter\floor{\lfloor}{\rfloor}

\hypersetup{
  colorlinks=true,      % colored text instead of boxed links
  pdfborder={0 0 0},    % belt-and-suspenders: no borders even if colorlinks is off
  linkcolor={blue},     % internal links (sections, theorems, refs)
  citecolor={blue},     % bibliography citations
  urlcolor={blue},      % external URLs
  breaklinks=true,      % allow links (URLs) to break across lines
}

\newcommand{\mathify}[1]{\ifmmode{#1}\else\mbox{$#1$}\fi}

\newcommand{\eps}{\varepsilon}

\newcommand{\ignore}[1]{}

\newcommand{\satisfying}[1]{\ensuremath{\mathsf{Sol}(#1)}}

\newcommand{\Cnt}[2]{\ensuremath{\mathsf{Cnt}_{\langle #1, #2 \rangle}}}
\newcommand{\good}{\mathsf{Good}}
\newcommand{\bad}{\mathsf{Bad}}

\newcommand{\hybrid}{\ensuremath{\mathsf{APS\mbox{-}Estimator}}}
\newcommand{\naive}{\ensuremath{\mathsf{Naive}}}
\newcommand{\hashing}{\ensuremath{\mathsf{HashingEstimator}}}

\newcommand{\Threshold}{\mathsf{thr}_0}
\newcommand{\Fail}{\mathsf{Fail}_{\mathsf{Sample}}}

\newcommand{\DistinctSample}{\ensuremath{\mathsf{DistinctSample}}}

\newcommand{\thresh}{\mathsf{thr}_1}

\newcommand{\Cov}{\mathsf{Cov} }

\newcommand{\streamsize}{M}
\newcommand{\universe}{\Omega_n}

\newcommand{\Hteop}{\ensuremath{\mathsf{H_{Teop}}}}

\newcommand{\range}{\ensuremath{\mathsf{Range}}}
\newcommand{\vol}{\ensuremath{\mathsf{Volume}}}

\newtheorem{theorem}{Theorem}[section]

\newtheorem{remark}[theorem]{Remark}
\newtheorem{definition}[theorem]{Definition}
\newtheorem{lemma}[theorem]{Lemma}
\newtheorem{claim}[theorem]{Claim}

\newtheorem{observation}[theorem]{Observation}

\newtheorem{problem}[theorem]{Problem}

\usepackage{bbm} %
\usepackage{mathtools}

\algnotext{EndFor}
\algnotext{EndIf}
\algnotext{EndWhile}
\algblockdefx{Do}{EndDo}{\textbf{do}\;}{}
\algnotext{EndDo}
\algnewcommand\algorithmicswitch{\textbf{switch}}
\algnewcommand\algorithmiccase{\textbf{case}}
\algdef{SE}[SWITCH]{Switch}{EndSwitch}[1]{\algorithmicswitch\ #1\ \algorithmicdo}{\algorithmicend\ \algorithmicswitch}%
\algdef{SE}[CASE]{Case}{EndCase}[1]{\algorithmiccase\ #1}{\algorithmicend\ \algorithmiccase}%
\algtext*{EndSwitch}%
\algtext*{EndCase}%

\newcommand{\rev}[1]{{\color{black}#1}}
\usepackage[colorinlistoftodos,prependcaption,textsize=tiny]{todonotes}

\title{Estimating Size of the Union of Sets in Streaming Model}
\titlenote{This is a significantly revised version of the paper that appeared in the 
proceedings of the 40th ACM SIGMOD-SIGACT-SIGAI Symposium on Principles of 
Database Systems (PODS-21). The main claims of the paper remain unchanged; 
however, we have fixed several typos and bugs in the proofs. The proofs and theorem
statements have also been formalized in Lean. 
The Lean formalization is available at 
\protect\url{https://github.com/meelgroup/delphic21lean}
}\titlenote{
The authors decided to forgo the old
convention of alphabetical ordering of
authors in favor of a randomized ordering,
denoted by \textcircled{r}. The publicly
verifiable record of the randomization is
available at \protect\url{https://www.aeaweb.org/journals/policies/random-author-order/search
}
}

\author{Kuldeep  S. Meel \textcircled{r} }
\affiliation{%
	\institution{National University of Singapore}
}

\author{N. V. Vinodchandran \textcircled{r}}
\affiliation{
	\institution{University of Nebraska, Lincoln}
}

\author{Sourav Chakraborty}
\affiliation{%
	\institution{Indian Statistical Institute, Kolkata}
}
\begin{document}
\fancyhead{}

\begin{abstract}
In this paper we study the problem of estimating the size of the union of sets
$S_1, \dots, S_M$, where each set $S_i \subseteq \Omega$ (for some discrete
universe $\Omega$) is implicitly presented and comes in a streaming fashion.
We define the notion of Delphic sets to capture a class of streaming problems
where membership, sampling, and counting calls to the sets are efficient. In
particular, we show that our notion of Delphic sets captures three well-known
problems: Klee's measure problem (discrete version), test coverage estimation,
and model counting of DNF formulas.
The Klee's measure problem is the problem of computing the volume of a union of
multi-dimensional axis-aligned rectangles. The test coverage estimation problem,
in the context of combinatorial testing, focuses on the computation of the
coverage measure for a given testing array. Combinatorial testing is a
fundamental technique in hardware and software testing. Finally, given a DNF
formula $\varphi = D_1 \vee D_2 \vee \ldots \vee D_M$, the model counting problem
seeks to compute the number of satisfying assignments of $\varphi$.

The primary contribution of our work is a simple and efficient sampling-based
algorithm, called {\hybrid}, for outputting an $(\varepsilon,\delta)$-approximation
of the cardinality of the union of (Delphic) sets in the streaming setting. Our
algorithm has space complexity $O(R\log |\Omega|)$ and update time complexity
$O(R\log R \cdot \log(M/\delta) \cdot \log|\Omega|)$, where
$R = O\!\left(\log(M/\delta)\cdot \varepsilon^{-2}\right)$.
Consequently, for the streaming Klee's measure problem, our algorithm provides
the first algorithm with update time complexity that linearly depends on the
dimension $d$ for $d>1$. This settles an open problem of Tirthapura and Woodruff
(PODS~12). Furthermore, a straightforward application of our algorithm results
in efficient algorithms for the coverage estimation problem and the DNF model
counting problem in the streaming setting. 

We then investigate whether the space complexity for coverage estimation can be
further improved, and in this context, we present another streaming algorithm
that uses near-optimal space complexity but uses an update algorithm that is in
${\rm P}^{\rm NP}$, thereby showcasing an interesting time--space trade-off in
the streaming setting.

It is worth remarking that we view a key strength of our work as the simplicity
of both the algorithm and its theoretical analysis, which makes it amenable to
practical implementation and easy adoption.

\end{abstract}

\maketitle

\section{Introduction}\label{sec:intro}

Estimating the size of the union of sets is a fundamental problem in computer
science. The goal, usually, is to design an efficient randomized algorithm that
can output an $(\varepsilon, \delta)$-approximation of the size of the union of
sets. We say that a random variable $Z$ is an $(\varepsilon,\delta)$-approximation
of $Y$ if $\Pr[|Z-Y| \leq \varepsilon |Y|] \geq 1-\delta$.

In this paper, we focus on estimating the union of sets in a streaming setting.
We consider a family of sets, which we call {\em Delphic sets}
(see Definition~\ref{def:delphic}), for which membership, sampling, and counting
queries can be implemented {\em efficiently}. To showcase the generality of
Delphic sets, we first present three problems arising in diverse domains that can
be captured by Delphic sets: Klee's measure problem, test coverage estimation,
and model counting for DNF formulas. The three problems are defined as follows:

\subsection*{Klee's Measure Problem}
We define the discrete version of Klee's measure problem (KMP) in the streaming
setting.
\begin{definition}\label{def:rectangle}
A $d$-dimensional axis-aligned rectangle $\mathbf{r}$ over a universe
$[\Delta]^d$, where $\Delta$ is a totally ordered discrete set, is
defined as $[a_1,b_1] \times [a_2,b_2] \times \ldots \times [a_d, b_d]$. Given a
rectangle $\mathbf{r}$, let $\range(\mathbf{r})$ denote the set of tuples
$\{(x_1,\ldots,x_d)\}$ where $a_i \leq x_i \leq b_i$ and $x_i$ is an integer.
Note that every $d$-dimensional rectangle can be succinctly represented by the
tuple $(a_1, b_1, \cdots, a_d, b_d)$. The streaming Klee's measure problem is the
following: given a stream $\mathcal{R}$ of size $\streamsize$ such that
$\mathcal{R} = \langle \mathbf{r}_1, \mathbf{r}_2, \cdots, \mathbf{r}_{\streamsize} \rangle$,
where each item $\mathbf{r}_i$ is a $d$-dimensional rectangle, output an
$(\varepsilon, \delta)$-approximation of the volume of $\mathcal{R}$, where
\begin{align*}
    \vol(\mathcal{R}) = \left| \bigcup_{1 \leq i \leq \streamsize} \range(\mathbf{r}_i) \right|.
\end{align*}
\end{definition}

\subsubsection*{Motivation}
Klee's measure problem (KMP) is well investigated in computational geometry.
Klee, in 1977, introduced the one-dimensional version of the problem over the
reals: given $n$ intervals in $\mathbb{R}$, compute the size of their
union~\cite{klee1977can}. Klee presented an $O(n\log n)$-time algorithm for the
problem in the traditional RAM model, which was later proved to be optimal by
Fredman and Weide~\cite{fredman1978complexity}. Since its introduction, KMP has
been studied extensively in computational geometry with the goal of designing
efficient algorithms in the traditional RAM model. As certain data objects in
databases can be represented by axis-parallel multi-dimensional rectangles, the
KMP problem is significant in databases~\cite{cai2000parametric,
lazaridis2001progressive,papadias2001efficient,tao2004range,zhang2008computing}.
In addition to computer science areas, algorithms for KMP have recently found
applications in a wide range of practical areas, including environmental
chemistry~\cite{CPGJ19} and lunar archaeology~\cite{BAMLM20}.

The discrete version of KMP that we consider in this paper is studied in the
streaming community as {\em range-efficient $\mathbb{F}_0$}
computation~\cite{PT07,TW12}. Other than its natural appeal, this problem is
interesting because several significant problems, including max-dominance
norm~\cite{cormode2003estimating}, counting triangles in
graphs~\cite{BKS02}, and the distinct summation
problem~\cite{considine2004approximate}, can be reduced to computing
$\mathbb{F}_0$ over multi-dimensional discrete ranges.

\subsection*{Test Coverage Estimation Problem}

\begin{definition}\label{def:cov}
For an $n$-bit binary string $\mathbf{a} = a_1a_2\cdots a_n \in \{0,1\}^n$, its
{\em $t$-coverage}, denoted by $\Cov_{t}(\mathbf{a})$, is defined as
\begin{align*}
 \Cov_{t}(\mathbf{a}) = & \left\{(T,\mathbf{y}) \ \mid\ T \subseteq [n], |T| = t,
 \mathbf{y} \in \{0,1\}^t \text{ and } \right. \\
 & \left. \text{ the restriction of } \mathbf{a} \text{ to indices in } T
 \text{ gives } \mathbf{y} \right\}.
\end{align*}
The streaming coverage estimation problem is defined as follows. Given a stream
$\mathcal{A}$ of size $\streamsize$ such that
$\mathcal{A} = \langle \mathbf{a}_1, \dots, \mathbf{a}_{\streamsize} \rangle$, where
$\mathbf{a}_i \in \{0,1\}^n$, output an $(\varepsilon, \delta)$-approximation of
$|\Cov_t(\mathcal{A})|$ for any given $t$ ($1 \leq t \leq n$), where
$\Cov_t(\mathcal{A}) = \bigcup_{1 \leq i \leq \streamsize} \Cov_t(\mathbf{a}_i)$.
\end{definition}

\subsubsection*{Motivation}
Over the past half-century, the widespread adoption of software in diverse areas
has necessitated the design of highly configurable systems wherein the end-user
can choose a configuration of interest by setting the desired values of the
configuration parameters. The space of configurations is often astronomical.
For example, the possible number of configuration options of an embedded Linux
for microcontrollers is $7.7 \times 10^{417}$~\cite{pettproduct}. Since it is not
feasible to examine the behavior of a System Under Test (SUT) for all possible
configurations, combinatorial testing has emerged as a dominant
paradigm~\cite{KKL13}. A configuration is specified by assigning values to the
configuration parameters\footnote{While techniques developed in our work extend
to scenarios wherein each of the parameters takes values in finite domains, for
simplicity of exposition, we will focus on the case where each of the parameters
takes a binary value.}. Motivated by the observation that most often the errors in
system results are due to the interaction of a relatively small number of
parameters, the combinatorial testing paradigm focuses on covering
$t$-combinations of parameters for a given $t$. Note that when each of the
parameters takes binary values, the total number of possible $t$-combinations of
parameters over a test of size $n$ is $\binom{n}{t} 2^t$.

In the critical area of software testing, there has been a long line of work on
the design of test suite generators~\cite{mandl1985orthogonal,
tatsumi1987test,cohen1997aetg,bryce2009density,nie2011survey,
thum2014classification,medeiros2016comparison}. The earliest works focused on the
design of the smallest test suites such that all the $\binom{n}{t} 2^t$
combinations are covered. The optimal constructions still require test suites of
size $\mathcal{O}(t2^{t} \log n)$, which is intractable for large enough $t$.
Consequently, one is often interested in achieving as high a coverage as possible
within a given budget. Moreover, such generators are heuristic-based and fail to
provide any rigorous guarantees on the quality of their generated test suites. In
this context, it is critical to rigorously estimate the coverage of a given test
suite. Note that in the context of binary parameters, a test can be specified as
a binary string $\mathbf{a} = a_1 a_2 \ldots a_n \in \{0,1\}^n$. For many practical
systems, $n$, the number of configurable parameters, is very large. Therefore,
it is not practical to store all the tests. In addition, it is useful to have
estimation methods that deal with a growing test suite. In particular, from a
practical perspective, one envisions a coverage estimator to be a {\em monitor}
with as small a resource overhead as possible. These issues lend themselves to
investigating the coverage estimation problem in a data streaming framework.

\subsection*{Model Counting for DNF}
\begin{definition}
Let $X$ be a set of $n$ Boolean variables. A literal is a variable or its
negation. A term is a conjunction of literals. A DNF formula $\varphi$ over $X$
is a disjunction of terms. For a DNF formula $\varphi$ over $n$ variables, let
$\satisfying{\varphi}$ denote the set of satisfying assignments of $\varphi$.
The streaming DNF model counting problem is the following: given a stream
$\Phi = \langle D_1, D_2, \dots, D_{\streamsize} \rangle$, where $D_i$ is a term over $n$
variables, output an $(\varepsilon, \delta)$-approximation of
$|\satisfying{\varphi}|$, where $\varphi$ is the DNF formula
$\varphi = D_1 \vee D_2 \vee \ldots \vee D_{\streamsize}$.
\end{definition}

\subsubsection*{Motivation}
\#DNF (model counting for DNF, also referred to as DNF counting, and often denoted
by \#DNF in the literature) is a fundamental problem in computer science with a
wide variety of applications. Dalvi and Suciu~\cite{DS07} showed that queries in
probabilistic databases reduce to \#DNF. Another important application of \#DNF
arises from the domain of network unreliability: given a graph $G = (V, E)$,
wherein each edge $e_i$ fails with probability $p_i$, we are interested in
computing the probability that $s$ and $t$ are disconnected. Karger's seminal
work reduces the computation of network unreliability to \#DNF, wherein each
term represents a min-cut~\cite{K01}. The past few years have witnessed a surge
in interest in designing efficient FPRAS techniques for
\#DNF~\cite{MSV17, MSV18}.

\subsection{Our Results}
We define the notion of {\em Delphic sets} (Definition~\ref{def:delphic}) and
formulate the problem of estimating the size of the union of these sets in the
streaming setting (Problem~\ref{prob:delphic}). This generic framework captures
many union-of-sets size estimation problems, including all three problems
mentioned above. Then, in Theorem~\ref{thm:unknownM}, we present an efficient
algorithm for the problem. The efficiency of our algorithm is measured in terms
of its worst-case space complexity and its worst-case per-item update time.
Finally, we show how our algorithm for Delphic sets yields new efficient
algorithms for all three problems mentioned above. The algorithm for KMP solves
an open problem from the literature.

\subsection*{Delphic Sets as a Unifying Model}

Let $\{\Omega_n\}_{n\geq 1}$ be a family of discrete sets indexed by $n$.
We define the notion of a Delphic family over $\{\Omega_n\}_{n\geq 1}$.

\begin{definition}\label{def:delphic}
Let $\mathcal{F} = \{\mathcal{F}_n\}_{n\geq 1}$, where
$\mathcal{F}_n \subseteq 2^{\Omega_n}$. We call $\mathcal{F}$ a {\em Delphic
family} if the following holds: for every $n$ and for every set
$S \subseteq \Omega_n$, the following queries can be performed in
$O(\log |\Omega_n|)$ time.
\begin{enumerate}
    \item Determine the size of the set $S$;
    \item Draw a uniform random sample from $S$; and
    \item Given any $x$, check whether $x \in S$.
\end{enumerate}
\end{definition}

Next, we define the following streaming problem.

\begin{problem}\label{prob:delphic}
Fix a Delphic family $\mathcal{F}$. Given $n \in \mathbb{N}$, $\varepsilon \in (0,1]$, $\delta \in (0,1]$,
and a stream 
$\mathcal{S} = \langle S_1, S_2, \ldots, S_{\streamsize} \rangle$, wherein each $S_i$ belongs
to $\mathcal{F}_n$, the task is to output an
$(\varepsilon,\delta)$-approximation of
$\left| \bigcup_{i=1}^{\streamsize} S_i \right|$ assuming
access only to the three queries listed in Definition~\ref{def:delphic}.
\end{problem}

The main contribution of this work is a new adaptive sampling-based algorithm
for the above problem (Problem~\ref{prob:delphic}), as stated in the following
theorem.

\begin{restatable}{theorem}{smalltime}
\label{thm:unknownM}
There is a streaming algorithm, which we call $\hybrid$, that solves
Problem~\ref{prob:delphic}. The algorithm has worst-case space complexity
$O(R \log |\universe|)$ and update time
$O(R \log R \cdot \log(\streamsize/\delta) \cdot \log |\universe|)$, where
$R = O\left(\log(\streamsize/\delta)\cdot \varepsilon^{-2}\right)$.
\end{restatable}

\subsection*{Klee's Measure Problem}
We first observe that Klee's measure problem can be formulated as Delphic coverage
by constructing the set $S_i = \range(\mathbf{r}_i)$ corresponding to each
$\mathbf{r}_i$ and observing that each such $S_i$ belongs to the Delphic family.
Therefore, the following corollary follows from
Theorem~\ref{thm:unknownM}.

\begin{restatable}{corollary}{kleecorr}
\label{corr:klee}
There is a streaming algorithm that given  a stream
$\mathcal{R} = \langle \mathbf{r}_1, \mathbf{r}_2, \cdots, \mathbf{r}_{\streamsize} \rangle$,
where each $\mathbf{r}_i$ is a $d$-dimensional rectangle over
$[\Delta]^d$, computes an $(\varepsilon,\delta)$-approximation of
$\vol(\mathcal{R})$. The algorithm has worst-case space
$O\left(d(\log \Delta) \cdot \log(\streamsize/\delta)\cdot \varepsilon^{-2}\right)$ and
worst-case update time complexity 
$$O\left(d(\log \Delta) \cdot (\log(\streamsize/\delta))^2 \cdot \right.\left. \log\log(\streamsize/\delta) \cdot \right.\left. 
\varepsilon^{-2}\log \varepsilon^{-1}\right).$$
\end{restatable}

\begin{remark}\label{remark}
Corollary~\ref{corr:klee} provides the first efficient algorithm with linear
dependence on $d$ for Klee's measure problem in the streaming model. This resolves
an open problem of Tirthapura and Woodruff~\cite{TW12}. We remark that Tirthapura
and Woodruff claimed an algorithm for KMP with space and update time complexity
$O\left(d(\log \Delta) \cdot (\log \streamsize) \cdot \varepsilon^{-2} \right.$ $\left. \cdot
\log(\frac{1}{\delta})\right)$. However, they later retracted this
claim~\cite{Woodruff20}. Their method only yields per-item worst-case update time
complexity\\ $\mathrm{poly}((\log \Delta)^d, \varepsilon^{-1},
\log \frac{1}{\delta})$.
\end{remark}

\subsubsection*{\bfseries Multi-Dimensional Arithmetic Progression}
We then focus on a generalization of Klee's measure problem by generalizing the
notion of a range $[a_i,b_i]$ to arithmetic progressions, which was studied
previously by Pavan and Tirthapura~\cite{PT07} for $d=1$. Let $[a, b, c]$
represent the arithmetic progression with common difference $c$ in the range
$[a, b]$, i.e., $a, a+c, a+2c, \ldots, a + jc$, where $j$ is the largest integer
such that $a + jc \leq b$. Consider a stream
$\mathcal{R} = \langle r_1, r_2, \ldots, r_m \rangle$, wherein each
$r_i = [a_1, b_1, c_1] \times \cdots \times [a_d, b_d, c_d]$. We generalize
$\range(r_i)$ to denote the set of tuples $\{(x_1, \ldots, x_d)\}$ where
$a_i \leq x_i \leq b_i$ and $x_i = a_i + k \cdot c_i$ for some non-negative
integer $k$. Similarly,
$\vol(\mathcal{R}) = \left| \bigcup_{i=1}^{m} \range(r_i) \right|$. By observing
that $\range(r_i)$ for $d$-dimensional arithmetic progressions belongs to the
Delphic family, we obtain the following streaming algorithm.

\begin{restatable}{corollary}{kleearith}\label{corr:kleearith}
There is a streaming algorithm that, given any real numbers
$\varepsilon, \delta < 1$, and a stream
$\mathcal{R} = \langle \mathbf{r}_1, \mathbf{r}_2, \cdots, \mathbf{r}_{\streamsize} \rangle$
consisting of $d$-dimensional arithmetic progressions, computes an
$(\varepsilon,\delta)$-approximation of $\vol(\mathcal{R})$. The algorithm has
worst-case space complexity
$O\left(d(\log \Delta) \cdot \log(\streamsize/\delta)\cdot \varepsilon^{-2}\right)$ and
worst-case update time complexity
$O\left(d(\log \Delta) \cdot (\log(\streamsize/\delta))^2 \log\log(\streamsize/\delta) \cdot
\varepsilon^{-2}\log \varepsilon^{-1}\right).$
\end{restatable}

Observe that Klee's measure problem is a special case of multi-dimensional
arithmetic progressions wherein $c_i = 1$. It is worth remarking that, in
comparison to prior work that focused on the special case of arithmetic
progressions for $d=1$, the algorithm for multi-dimensional ranges (i.e.,
Klee's measure problem) can be directly lifted to multi-dimensional arithmetic
progressions, and the space and time complexity does not change.

\subsection*{Test Coverage Estimation}
Next, we observe that the test coverage estimation problem can also be formulated
as Delphic coverage by constructing $S_i = \Cov_t(\mathbf{a}_i)$ for each
$\mathbf{a}_i$, and again observing that each such $S_i$ belongs to the Delphic
family.

\begin{restatable}{corollary}{covercorr}
\label{corr:cov}
There is a streaming algorithm $\hybrid$ that, given any real numbers
$\varepsilon, \delta < 1$, and a stream
$\mathcal{A} = \langle \mathbf{a}_1, \cdots, \mathbf{a}_{\streamsize} \rangle$, where
$\mathbf{a}_i \in \{0,1\}^n$, computes an $(\varepsilon,\delta)$-approximation of
$|\Cov_t(\mathcal{A})|$. The algorithm has space complexity
$O\left(t(\log n) \cdot \log(\streamsize/\delta)\cdot \varepsilon^{-2}\right)$ and
worst-case update time complexity
$O\left(t(\log n) \cdot (\log(\streamsize/\delta))^2 \right.$ $\left.\log\log(\streamsize/\delta) \cdot
\varepsilon^{-2}\log \varepsilon^{-1}\right)$.
\end{restatable}

We also investigate whether the space complexity can be further improved in the
context of test coverage estimation. We present a hashing-based algorithm that
trades an improvement in space complexity for an increase in the update-time
overhead by relying on the use of NP oracles.

\begin{restatable}{theorem}{smallspacewithnp}
\label{thm:smallspace}
There is a streaming algorithm $\hashing$ with access to an {\rm NP} oracle that,
given a stream $\mathcal{A} = \langle \mathbf{a}_1, \cdots, \mathbf{a}_{\streamsize} \rangle$
and real numbers $0<\varepsilon, \delta<1$, where each
$\mathbf{a}_i \in \{0,1\}^n$, computes an $(\varepsilon,\delta)$-approximation of
$|\Cov_t(\mathcal{A})|$. The algorithm uses
$O\left(t \log n \cdot \varepsilon^{-2} \cdot \log\frac{1}{\delta}\right)$ space
and ${\rm poly}(n,t,1/\varepsilon)$ update time. 
\end{restatable}

For an oracle algorithm, we only count the space used by the base algorithm. 
The space complexity of the above algorithm, in general, is tight up to an
$O(\log n)$ factor, as for $t = n$ the problem reduces to $\mathbb{F}_0$
computation in the traditional insertion-only data stream model, for which a
lower bound of $\Omega(n)$ holds (for constant $\varepsilon,\delta$) - note that in
our case items are from $\{0,1\}^n$. 
The problem of designing an algorithm for the 
test coverage estimation problem that achieves tight bounds from both space and time complexity
perspectives remains open.

\subsection*{DNF Counting}
We again observe that DNF counting can be formulated as Delphic coverage by
constructing the set $S_i = \satisfying{D_i}$ corresponding to each term $D_i$,
and observing that each such $S_i$ belongs to the Delphic family. Therefore, the
following corollary follows from Theorem~\ref{thm:unknownM}.

\begin{restatable}{corollary}{dnfcorr}\label{corr:dnf}
There is a streaming algorithm that, given any real numbers
$\varepsilon, \delta < 1$, and a stream
$\langle D_1, D_2, \ldots, D_{\streamsize} \rangle$ of terms over $n$ variables, computes an
$(\varepsilon,\delta)$-approximation of $|\satisfying{\varphi}|$, where
$\varphi = D_1 \vee D_2 \vee \ldots \vee D_{\streamsize}$. The algorithm takes
$O\left(n \cdot \log(\streamsize/\delta) \cdot \varepsilon^{-2}\right)$ space and
$O\left(n \cdot (\log(\streamsize/\delta))^2 \right.$ $\left. \log\log(\streamsize/\delta) \cdot
\varepsilon^{-2}\log \varepsilon^{-1}\right)$ update time.
\end{restatable}

\subsection{Techniques}\label{sec:technique}

Our main algorithm, $\hybrid$, is a simple sampling-based algorithm for estimating
$\bigl|\bigcup_i S_i\bigr|$ in a streaming setting.
At a high level, if it were possible to sample every element of
$\bigcup_i S_i$ independently and identically with some fixed probability $p$,
then the estimator $|\mathcal{X}|/p$ would be an unbiased estimate of the union
size, and standard concentration bounds would yield $(\varepsilon,\delta)$-guarantees
for an appropriate choice of $p$ (on the order of
$1/(\varepsilon^2 \cdot |\bigcup_i S_i|)$).

The difficulty is that neither the union size nor a suitable sampling probability
$p$ is known in advance, and the sets $S_i$ arrive one by one as a stream.
We therefore begin by explaining the intuition assuming a fixed sampling
probability, and then explain how the algorithm adapts when $p$ itself evolves
during execution.

Suppose, for the moment, that a fixed sampling probability $p$ were known.
The algorithm maintains a bucket $\mathcal{X}$ with the informal invariant that,
after processing $S_j$, every element of $\bigcup_{i=1}^j S_i$ is included in
$\mathcal{X}$ independently with probability $p$.
When $S_1$ arrives, this can be achieved by sampling
$N_1 \sim \mathsf{Bin}(|S_1|,p)$ and storing $N_1$ uniformly chosen distinct
elements of $S_1$.

Now assume this intuition holds after processing $S_j$, and consider the arrival
of $S_{j+1}$. Simply sampling elements of $S_{j+1}$ with probability $p$ would bias
the sketch toward elements appearing in many sets.
To avoid this, the algorithm first removes from $\mathcal{X}$ all elements that
belong to $S_{j+1}$, and then samples elements of $S_{j+1}$ with probability $p$
and inserts them into $\mathcal{X}$. This remove--then--resample step ensures that,
for any element $x \in S_{j+1}$, the probability that $x$ belongs to $\mathcal{X}$
after processing $S_{j+1}$ is exactly $p$, regardless of whether $x$ appeared in
earlier sets.

The remaining issue is how to choose $p$ without knowing the union size in advance.
The algorithm addresses this by starting with $p=1$ and dynamically decreasing it
as needed. Whenever the sketch $\mathcal{X}$ grows beyond a prescribed threshold
$\Threshold$, the algorithm removes each element of $\mathcal{X}$ independently
with probability $1/2$ and updates $p \gets p/2$. Intuitively, this thinning step
keeps the sketch size under control while maintaining the correct sampling scale.

It is important to stress that this discussion is intentionally informal.
Once thinning steps are introduced, $p$ becomes a random variable whose value
depends on the execution of the algorithm, and it is no longer meaningful to
literally claim that elements are sampled independently with probability $p$
throughout the run.  The formal analysis, presented in the proof of correctness,
carefully accounts for the randomness and adaptivity of $p$ and shows that the
estimator $|\mathcal{X}|/p$ nevertheless provides the desired
$(\varepsilon,\delta)$-approximation.

For generality, our time complexity analysis relies only on the three black-box
operations provided by Delphic sets: membership testing, size queries, and access
to random samples. In particular, generating distinct samples is handled using
only this abstract interface. In several concrete settings—such as
$\Cov(\mathbf{a}_i)$, $\range(r_i)$, or $\satisfying{D_i}$—more efficient sampling
procedures are possible, leading to improved update times. We leave a refined,
problem-specific analysis to future work.

A key strength of our approach is the simplicity of both the algorithm $\hybrid$
and its analysis. The algorithm consists of only a small number of intuitive
operations, and the proof relies on elementary probabilistic arguments. This
simplicity makes the method easy to implement and facilitates independent
verification and reuse in other streaming and sampling contexts.

\paragraph{Organization}
The rest of the paper is organized as follows: We discuss related work in Section~\ref{sec:related}. We then discuss notations and preliminaries in Section~\ref{sec:notations}. The paper's primary technical contribution is presented in Section~\ref{sec:algorithms}, which consists of four parts. Section~\ref{sec:delphic} provides the proof of the main technical theorem~\ref{thm:unknownM}. We then demonstrate Theorem~\ref{thm:unknownM} can be applied in different contexts, Klee's Measure Problem (Section~\ref{sec:klee}), coverage estimation (Section~\ref{sec:coverage}), and DNF counting (Section~\ref{sec:dnf}) to obtain efficient streaming algorithms. We finally conclude in Section~\ref{sec:conclusion}.

\section{Related Work}\label{sec:related}

Starting with the seminar work of Alon, Matias, and Szegedy~\cite{AMS99}, the streaming model of computation has emerged as an important area of research in theoretical computer science. This model is well suited to investigate algorithmic problems that arise from real life situations dealing with large data. A central focus of investigation has been on estimating the frequency moments $\mathbb{F}_k$ of a stream of data items. 
In particular, considerable work has been done in designing algorithms for estimating the the $0^{th}$ frequency moment ($\mathbb{F}_0$), the number of distinct elements in the
stream, culminating in the development of an algorithm with 
optimal space complexity $O(\log |\Omega| + \frac{1}{\varepsilon^2} )$ and $O(1)$ update time~\cite{KNW10}, where $\Omega$ is the universe. 

In the case when items in the data stream succinctly represent a set of elements of the universe, $\mathbb{F}_0$ estimation becomes estimation of size of the {\em union of sets}. The union-of-sets problem has been studied in  approximate counting literature. The line of work that is closest to ours is the work on DNF counting problem initiated by Karp and Luby~\cite{KL83}. Since the work of Karp and Luby, substantial research has gone into understanding various aspects of DNF counting problem including designing hashing-based algorithms~\cite{KLM89,DKLR00,MSV17,MSV18,MSV19,CMV16}. We note that Karp-Luby algorithm can be adapted to counting union of sets in the streaming setting to get an algorithm with space and time complexity $O(\frac{M \log|\Omega|}{\epsilon^2}\log M\log n)$. In comparison, we achieve only a logarithmic dependence on $M$. 

As mentioned in the introduction, Klee's Measure Problem (KMP) is a fundamental problem that is well investigated in computational geometry with the  focus of designing efficient algorithms in the traditional RAM model. Klee introduced the one-dimensional version of the problem over reals and presented $O(n\log n)$ time algorithm where $n$ is the number of line segments~\cite{klee1977can}. Fredman and Weide showed that this is optimal in time (under certain model)~\cite{fredman1978complexity}. Since then substantial work has gone into extending the algorithms to multidimensional case~\cite{bentley1977algorithms,overmars1991new,chan2010slightly,chen2005space, BringmannF10, GudmundssonP17, chlebus1998klee} and also with space complexity considerations~\cite{chen2005space,vahrenhold2007place}. 

Discrete version of KMP over streaming model has been considered before. However the success has been limited~\cite{BKS02, PT07, TW12,SP09}. Pavan and Tirthapura
considered the problem for one dimensional ranges over the discrete domain $\{1,\ldots,n\}$ and gave an algorithm with $O(\varepsilon^{-2}\log n\log 1/\delta )$ 
space and $O(\log {n/\varepsilon} \log 1/\delta)$ update time~\cite{PT07}. Sharma, Busch, Vaidyanathan, Rai, and Trahan considered the
two-dimensional version but only gave a $O(\sqrt{\log U})$-approximation for the general case where $U$ is the total number of
discrete points in the space~\cite{SHARMA2015688}. Tirthapura and Woodruff~\cite{TW12} considered the general $d$ dimensional problem. They
presented an algorithm, based of range efficient implementations of {\em count sketch} algorithm~\cite{CCF04}  and recursive sketches~\cite{IW05,BO10}, which is efficient in space complexity. However the update time of the algorithm has exponential dependency
on the dimension $d$ (see Remark~\ref{remark}). In a concurrent work, 
Pavan \textcircled{r} Vinodchandran \textcircled{r} Bhattacharyya \textcircled{r} Meel\footnote{ \textcircled{r} refers to the randomized author ordering.} also proposed another hashing-based technique with exponential  dependence on the dimension $d$.  

\ignore{

Estimating the size of the union of sets $S_1, \dots, S_M$  is an well studied problem in algorithms. Usually the complexity measure is the number of elements whose presence or absence in the sets has to be noted during the run of the algorithms. Karp-Luby \cite{KL83} and later Karp-Luby-Madras \cite{KLM89} showed that the size of the union of the sets can be $(\epsilon, \delta)$-estimated by keeping track of $O(\frac{M \log |\Omega|}{\epsilon^2}\log(1/\delta))$ number of elements.

\subsubsection{Time vs Space trade-off in streaming complexity}

The streaming model of computing a function of the input is a very well studied area in theoretical computer science has also has applications in many real life settings. The paper of Alon-Matias-Szegedy \cite{AMS99}
was one of the seminal papers in this area. They gave non-trivial bound for estimating the number of distinct elements in the stream. Recently Kane, Nelson, and Woodruff~\cite{KNW10} present tight upper bounds for this problem. The union of sets can be estimated by using their algorithms. The space complexity is optimal, but the update time for each set $S_i$ is $O(|S_i|)$, which in the Coverage-Estimation problem in $\binom{n}{t}$. Our algorithm in Theorem~\ref{thm:smallspace} uses almost optimal space complexity but the time complexity is polynomial assuming we have access to an NP oracle. 
}

\section{Notations and Preliminaries}\label{sec:notations}
We will denote by $[n]$ the set $\{1, 2, \dots, n\}$ and by $\binom{[n]}{t}$ the set of all subsets of 
$[n]$ of size $t$. For any $n\in \mathbb{N}$ and any $p \in [0,1]$ we will also use $\mathsf{Bin}(n, p)$ to denote the Binomial distribution over the set of natural numbers $\{0, \dots, n\}$ where probability of a number $0\leq m\leq n$ is $\binom{n}{m}p^m(1-p)^{n-m}$.

At any point the input item is a length $n$ string. However, as done in the case of traditional space bounded computations, for counting space, we will not include the space required to represent the input item. We will consider that input is available on a read-only input tape that allows random access and does not contribute to the space used by the algorithm. In this paper, we consider unit-cost model and assume all basic operations including arithmetic operations on words can be performed in unit time. 

Our technical analysis employs the standard concentration inequalities: Chernoff bound, Chebyshev's bound, and Paley-Zygmund Inequality.   For the proof of Theorem~\ref{thm:unknownM} we will need the following theorem, popularly known as the Coupon Collector Problem.

\begin{theorem}\label{thm:cc}
Given access to uniform random samples from a set $T$ and a number $r\leq |T|$, let $Z_r$ be a random variable that stand for the number of independent uniform random samples from $T$ needed before we get $r$ distinct samples from $T$. Then 
$$\Pr\left[Z_r > \beta r\log r\right] \leq r^{-\beta +1}.$$
\end{theorem}

\subsection*{Pairwise Independent Hash functions}
Let $n,m\in \mathbb{N}$ and $\mathcal{H}(n,m) \triangleq \{ h:\{0,1\}^{n} \rightarrow \{0,1\}^m \}$ be a family of hash functions mapping $\{0,1\}^n$ to $\{0,1\}^m$. We use $h \xleftarrow{R} \mathcal{H}(n,m)$ to denote the probability space obtained by choosing a function $h$ uniformly at random from $\mathcal{H}(n,m)$. 
\begin{definition}
	A family of hash functions $\mathcal{H}(n,m)$ is $2-$wise independent if $\forall \alpha_1, \alpha_2 \in \{0,1\}^m$,  $\text{ distinct } x_1, x_2,  \in \{0,1\}^n, h \xleftarrow{R} \mathcal{H}(n,m)$, 
	\begin{align}
	\Pr[(h(x_1) = \alpha_1) \wedge (h(x_2) = \alpha_2) ] = \frac{1}{2^{2m}} %
	\end{align}
\end{definition}

\paragraph{Explicit families}	
In this work, one hash family of particular interest is $\Hteop(n,m)$, which is known to be  2-wise independent~\cite{carter1977universal}. The family is defined as follows:
$\Hteop(n,m) \triangleq \{ h: \{0,1\}^n \rightarrow \{0,1\}^m  \}$ is the family of functions of the form $h(x) = Ax+b$ with $A \in \mathbb{F}_{2}^{m \times n}$ and $b \in \mathbb{F}_{2}^{m \times 1}$ where $A$ is a uniformly randomly chosen Toeplitz matrix of size $m\times n$ while $b$ is uniformly randomly  matrix of size $m \times 1$. it is worth noticing that $\Hteop$ can be represented with $\Theta(n+m)$-bits. 
	For every $m \in \{1, \ldots
	n\}$, the $m^{th}$ prefix-slice of $h$, denoted $h_{m}$, is a
	map from $\{0,1\}^{n}$ to $\{0,1\}^m$, where $h_{m}(y)$ is the first $m$ bits of $h(y)$. Observe that when $h(x) = Ax+b$,  $h_{m}(x) = A_{m}x+b_{m}$, where $A_{m}$ denotes the submatrix formed by the first $m$ rows of $A$ and $b_{m}$ is the first $m$ entries of the vector $b$.

\section{The Algorithms}\label{sec:algorithms}

\subsection{Coverage of Delphic Sets}\label{sec:delphic}

We restate the theorem that we prove in this section.

\smalltime*

\begin{algorithm}
	\caption{$\hybrid$}\label{algo:final}
	\begin{algorithmic}[1]
		\State Initialize $p \gets 1$; 
         $\Threshold \gets \max\left(  12  \cdot \frac{\ln (48/\delta)}{\varepsilon^2}, 6 (\ln \frac{6}{\delta} + \ln \streamsize) \right)$
		\State Initialize $\mathcal{X} \gets \emptyset$
		
		\For{$i = 1$ to $\streamsize$}
		
		\For{all $s\in \mathcal{X}$}\label{line:final-for-sketch-begin}
		\If{$s\in S_i$}
		\State remove $s$ from $\mathcal{X}$\label{line:final-remove}
		\EndIf
		\EndFor\label{line:final-for-sketch-end}
		
		\State $N_i \gets \mathsf{Bin}(|S_i|,p)$\label{line:final-pick}
		
		\While{$N_i + |\mathcal{X}| \geq \Threshold$}\label{line:final-check-threshold-begin}
		\State $N_i \gets \mathsf{Bin}(N_i,1/2)$
		\State Throw away each element of $\mathcal{X}$ with probability $\frac{1}{2}$ \label{line:final-throw}
		\State $p = p/2$\label{line:final-check-threshold-end}
		\EndWhile
        \State $\mathcal{Z} \gets $\DistinctSample$(S_i,N_i)$ \label{line:sample-distinct}
        \State $\mathcal{X} \gets \mathcal{X} \cup \mathcal{Z}$ \label{line:adding-samples}
	\EndFor
		\State Output $\frac{|\mathcal{X}|}{p}$
	\end{algorithmic}
\end{algorithm}

\begin{algorithm}
    \caption{{\DistinctSample}($S_i$,$N_i$)}\label{algo:distinctsample}
    \begin{algorithmic}[1]
	\State Initialize $\ell \gets 0; k \gets 0$; $\mathcal{Z} \gets \emptyset$ \label{line:CCbegin}
		\While{$\ell < 1+ N_i\log N_i\log (\frac{6\streamsize}{\delta})$}\label{line:final-sample-loop-begin}
   \State $\ell \gets \ell+1$    
		\State Draw a random sample $y$ from $S_i$\label{line:final-sample} 
		 \If{$y\not\in \mathcal{Z}$} 
		 \State Add $y$ to $\mathcal{Z}$; $k \gets k+1$  
		\EndIf       
         \If{$k = N_i$} {\Return $\mathcal{Z}$}  \EndIf 
  \EndWhile  \label{line:final-sample-loop-end}
  \If  {$\ell =1+  N_i\log N_i\log \left(\frac{6\streamsize}{\delta}\right)$} \Return $\emptyset$ \label{line:ccfail}
  \EndIf 

    \end{algorithmic}
\end{algorithm}

\begin{algorithm}	\caption{$\naive$}\label{algo:s2}
	\begin{algorithmic}[1]
		\State Initialize $p \gets 1$ 
		\State  $\Threshold \gets \max\left(  12  \cdot \frac{\ln (48/\delta)}{\varepsilon^2}, 6 (\ln \frac{6}{\delta} + \ln \streamsize) \right)$
		\State Initialize $\mathcal{X} \gets \emptyset$
		
		\For{$i = 1$ to $\streamsize$}
		
		\For{all $s\in \mathcal{X}$}\label{line:s2final-for-sketch-begin}
		\If{$s\in S_i$}
		\State remove $s$ from $\mathcal{X}$\label{line:s2final-remove}
		\EndIf
		\EndFor\label{line:s2final-for-sketch-end}
		\State $N_i \gets \mathsf{Bin}(|S_i|,p)$\label{line:s2final-pick}
         {\color{black}  \State  Draw $N_i$ distinct samples from $S_i$ and add them to $\mathcal{X}$} \label{line:s2Nipick} 

{\color{black}		\While{ $|\mathcal{X}| \geq \Threshold$}\label{line:s2final-check-threshold-begin}
		\State Throw away each element of $\mathcal{X}$ with probability $\frac{1}{2}$ \label{line:s2final-throw} 
\State $p = p/2$\label{line:s2final-check-threshold-end}
  \EndWhile}   
\EndFor
\State Output $\frac{|\mathcal{X}|}{p}$
\end{algorithmic}
\end{algorithm}

\begin{proof}
We present the algorithm {$\hybrid$} in Algorithm~\ref{algo:final}. To analyze its correctness, we compare it to an
idealized variant {\naive}, which makes the following modifications:
\begin{itemize}
    \item Instead of invoking $\DistinctSample$, {\naive} assumes access to an oracle that always returns
    $N_i$ distinct samples from $S_i$.
    \item Instead of postponing sampling until $N_i+|\mathcal{X}|<\Threshold$, {\naive} samples $N_i$ elements
    immediately and, whenever $p$ is halved, removes each element of $\mathcal{X}$ independently with probability
    $1/2$.
\end{itemize}
We give a pseudo-code for {\naive} in Algorithm~\ref{algo:s2}.

Let $\mathsf{Bad}_{\naive}$ and $\mathsf{Bad}_{\hybrid}$ denote the events that the final output
$\frac{|\mathcal{X}|}{p}$ lies outside
$[(1-\eps)|S^{(\streamsize)}|,\,(1+\eps)|S^{(\streamsize)}|]$ for {\naive} and {\hybrid}, respectively.
Let $\Fail$ be the event that some call to $\DistinctSample$ in {\hybrid} returns $\emptyset$
(only counted when $N_i\ge 2$).

\begin{observation}\label{obs:hybrid-vs-naive}
\begin{align*}
\Pr[\mathsf{Bad}_{\hybrid}] \le \Pr[\mathsf{Bad}_{\naive}] + \Pr[\Fail].
\end{align*}
\end{observation}

\noindent\emph{Justification.}
We couple the randomness of {\hybrid} and {\naive} so that they use the same random bits for generating the values
$N_i$, for deciding when $p$ is decreased, and for thinning the sketch.
Conditioned on $\overline{\Fail}$, each call to $\DistinctSample$ returns exactly $N_i$ distinct elements, and the
postponed sampling in {\hybrid} produces exactly the same surviving elements that remain after the loop~\ref{line:s2final-for-sketch-begin}--~\ref{line:s2final-for-sketch-end}).  
{\naive}
Hence, under this coupling, the two algorithms produce identical outputs whenever $\overline{\Fail}$
occurs. Therefore,
\begin{align*}
\mathsf{Bad}_{\hybrid} \subseteq \Fail \cup \mathsf{Bad}_{\naive},
\end{align*}
which implies Observation~\ref{obs:hybrid-vs-naive}.

% \rev{The pathwise inclusion above is a convenient informal reading; the mechanism that
% holds rigorously is stochastic domination rather than a coupling under which the two runs
% literally coincide. In the machine-checked development both {\hybrid} and {\naive} are
% Markov chains, and the transition law of {\hybrid} is dominated, step by step, by the single
% kernel that {\naive} follows exactly; this domination propagates to the printed value
% $|\mathcal{X}|/p$, and the mass it discards at each step is accounted for by $\Pr[\Fail]$,
% which together give Observation~\ref{obs:hybrid-vs-naive}. The corresponding \emph{equality}
% of the one-step laws is in fact false: for any finite sampling budget, {\DistinctSample}
% inserts fewer than $N_i$ distinct elements with positive probability while the idealized step
% always inserts exactly $N_i$, so the two one-step distributions genuinely differ. The
% inequality of the Observation is thus the sharp statement.}

We now bound the two probabilities on the right-hand side.

\subsubsection*{Bounding $\Pr[\mathsf{Bad}_{\naive}]$}
For $i\in[\streamsize]$ and $j\ge 0$, let $Y_j^{(i)}$ denote the random subset of
$S^{(i)} = S_1\cup\cdots\cup S_i$ obtained by including each element independently with probability $2^{-j}$.
Algorithm~\ref{algo:s2} maintains the invariant that after processing $S_i$, we have
$\mathcal{X}=Y_j^{(i)}$ where $p=2^{-j}$.

For $1\le i\le \streamsize$ and $j\ge 0$, let $E_j^{(i)}$ be the event that after processing $S_i$ the value of $p$ is $2^{-j}$,
and let $A_j^{(i)}$ be the event that
\begin{align*}
|Y_j^{(i)}|
\notin
\Big[|S^{(i)}|2^{-j}(1-\eps),\,|S^{(i)}|2^{-j}(1+\eps)\Big].
\end{align*}
If $E_j^{(\streamsize)}$ occurs and $A_j^{(\streamsize)}$ does not, then
$\frac{|\mathcal{X}|}{p} \in [(1-\eps)|S^{(\streamsize)}|,(1+\eps)|S^{(\streamsize)}|]$. So
\begin{align*}
\Pr[\mathsf{Bad}_{\naive}]
\le
\Pr\left[\bigcup_{j\ge 0}\left(E_j^{(\streamsize)}\wedge A_j^{(\streamsize)}\right)\right].
\end{align*}

Let $j^*$ be the smallest integer such that
\begin{align*}
2^{-j^*} < \frac{\Threshold}{4|S^{(\streamsize)}|}.
\end{align*}
By a union bound,
\begin{align*}
\Pr[\mathsf{Bad}_{\naive}]
\le
\sum_{j=0}^{j^*-1}\Pr[A_j^{(\streamsize)}]
+
\Pr\left[\bigcup_{j\ge j^*}E_j^{(\streamsize)}\right].
\end{align*}

For $j\le j^*-1$, the expectation of $|Y_j^{(\streamsize)}|$ is $\rev{|S^{(\streamsize)}|}/2^j$ which is at least $2^{(j^*-1-j)}\frac{\Threshold}{4}$.
Thus by Chernoff bounds,
\begin{align*}
\Pr[A_j^{(\streamsize)}]
<
4\exp\left(-\frac{\eps^2\cdot 2^{(j^*-j-1)}\Threshold}{12}\right)
\le 4\left(\frac{\delta}{48}\right)^{2^{j^*-1-j}},
\end{align*}
using $\Threshold \ge 12\ln(48/\delta)/\eps^2$. Thus, 
\begin{align*}
\sum_{j=0}^{j^*-1}\Pr[A_j^{(\streamsize)}]
< \sum_{j=0}^{j^*-1} 4\left(\frac{\delta}{48}\right)^{2^{j^*-1-j}} < 
4\frac{\delta}{24} = \frac{\delta}{6}.
\end{align*}

The event $\bigcup_{j\ge j^*}E_j^{(\streamsize)}$ can occur only if, for some $i$,
$|Y_{j^*-1}^{(i)}|\ge\Threshold$ (note that 
$\mathbb{E}[|Y_{j^*-1}^{(i)}|]<\Threshold/2$). Thus from Chernoff bound, 
$$\rev{\Pr[|Y_{j^*-1}^{(i)}| \ge \Threshold]} < \Pr\big[\,\big|\,\rev{|Y_{j^*-1}^{(i)}|} - \mathbb{E}[\rev{|Y_{j^*-1}^{(i)}|}]\,\big| \geq \tfrac{\Threshold}{2}\,\big] <2\exp\left(-\frac{\Threshold}{6}\right) $$
Finally, using union bound over $i\in[\streamsize]$ we obtain
\begin{align*}
\Pr\left[\bigcup_{j\ge j^*}E_j^{(\streamsize)}\right]\le \frac{\delta}{6},
\end{align*}
using $\Threshold \ge 6(\ln(6/\delta)+\ln \streamsize)$.
Hence,
\begin{align*}
\Pr[\mathsf{Bad}_{\naive}] \le \frac{\delta}{3}.
\end{align*}

\subsubsection*{Bounding $\Pr[\Fail]$}
Fix an iteration $i$ and condition on $N_i=r$.
If $r\in\{0,1\}$, $\DistinctSample$ cannot fail.
If $r\ge 2$, then $\DistinctSample$ fails only if more than
$r\log r\log(6\streamsize/\delta)$ samples are needed to obtain $r$ distinct elements from $S_i$.
By the coupon collector bound (Theorem~\ref{thm:cc}),
\begin{align*}
\Pr[\text{{\DistinctSample} fails at iteration $i$}] \le \frac{2\delta}{3\streamsize}.
\end{align*}
A union bound over $i\in[\streamsize]$ yields
\begin{align*}
\Pr[\Fail] \le \frac{2\delta}{3}.
\end{align*}

\paragraph{Correctness of {\hybrid}.}
By Observation~\ref{obs:hybrid-vs-naive},
\begin{align*}
\Pr[\mathsf{Bad}_{\hybrid}]
\le
\frac{\delta}{3}+\frac{2\delta}{3}
=
\delta.
\end{align*}
Thus {\hybrid} outputs an $(\eps,\delta)$-estimate of $|S^{(\streamsize)}|$.

\paragraph{Space and Update Time Complexity.}
At all times $|\mathcal{X}|<\Threshold$, so the space usage is
$O(\Threshold\log|\universe|)$.
Each update performs at most $\Threshold$ membership queries and
$O(\Threshold\log\Threshold\log(3\streamsize/\delta))$ sampling queries, each costing
$O(\log|\universe|)$ time, yielding the stated bounds.
\end{proof}

\subsection{Klee's Measure Problem}\label{sec:klee}

We return to Klee's measure problem in the streaming setting and show that a
direct application of Theorem~\ref{thm:unknownM} yields the first algorithm with
space and update-time complexity polynomial in the dimension $d$. As a first step,
we show that the set of integer points induced by a rectangle is Delphic.

\begin{lemma}
For every $d$-dimensional axis-aligned rectangle $\mathbf{r}$,
the set $\range(\mathbf{r})$ belongs to the Delphic family.
\end{lemma}
\begin{proof}
Let $\mathbf{r} = [a_1,b_1] \times \cdots \times [a_d,b_d]$ be a rectangle over
$\Omega = [\Delta]^d$. Note that $\range(\mathbf{r}) \subseteq \Omega$.
\begin{enumerate}
    \item The size of $\range(\mathbf{r})$ is
    $\prod_{i=1}^{d} (b_i - a_i+1)$, which can be computed in $O(d)$ time.
    \item To draw a uniform random sample
    $x = (x_1,\ldots,x_d) \in \range(\mathbf{r})$, we independently sample each
    coordinate $x_i$ uniformly from the interval $[a_i,b_i]$, which can be done
    in $O(d \log \Delta)$ time.
    \item Given $x = (x_1,\ldots,x_d)$, we can test whether
    $x \in \range(\mathbf{r})$ by checking that $a_i \le x_i \le b_i$ for all
    $i \in [d]$, which takes $O(d)$ time.
\end{enumerate}
All three operations are therefore supported efficiently, and
$\range(\mathbf{r})$ is Delphic.
\end{proof}

Since each rectangle $\mathbf{r}_i$ in the stream implicitly represents the set
$\range(\mathbf{r}_i)$, the following corollary follows immediately from
Theorem~\ref{thm:unknownM}.

\kleecorr*

The notion of range $[a_i,b_i]$ can be generalized to arithmetic progressions, which was studied previously by Pavan and Tirthapura
for $d=1$. Our sampling model allows us to derive streaming algorithms for a more general model comprising of $d$-dimensional
arithmetic progressions: Let $[a, b, c]$ represent the arithmetic progression with common difference $c$ in the range $[a, b]$, i.e., $a , a+c, a+2c, a + jc$, where $j$ is the largest integer such that $a+ jc \leq b$.  Consider a stream $\mathcal{R} = \langle\textbf{r}_1, \textbf{r}_2, \ldots \textbf{r}_M\rangle$ wherein each $\textbf{r}_i = [a_1, b_1, c_1] \times \cdots  \times [a_d, b_d, c_d]$. We generalize $\range(\textbf{r})$ to denote the set of tuples $\{(x_1, \ldots x_d)\}$ where $a_i \leq x_i \leq b_i$ and $x_i = a_i + k\cdot c_i$ for some non-negative integer k. Similarly, $\vol(\mathcal{R}) = | \bigcup_{i=1}^{m} \range(\textbf{r}_i)|$. To apply Theorem~\ref{thm:unknownM}, we first show that $\range(\textbf{r})$ of a  $d$-dimensional arithmetic progressions is Delphic. 

\begin{lemma}
For each $d$-dimensional arithmetic progression $\textbf{r}$,  the set $\range(\textbf{r})$ belongs to Delphic family. 
\end{lemma}
\begin{proof}
Note that $\range(\textbf{r}) \subseteq [\Delta]^{d}$
\begin{enumerate}
    \item For a given $\textbf{r}$, the size of the set is simply $\prod_{i=1}^{d} \left( \floor{\frac{b_i - a_i}{c_i}} + 1\right)$, which can be computed in $O(d)$ time. 
    \item To draw a uniform sample $(x_1, \ldots x_d)$, $x_i$ is sampled uniformly at random from $[a_i,b_i, c_i]$. Note that to sample from  $[a_i,b_i, c_i]$, we first draw a number $k$ uniformly at random from $(\floor{\frac{b_i - a_i}{c_i}} +1)$, and then return $a_i + k\cdot c_i$
    \item Given $x =  (x_1, \ldots x_d)$, we can check if $ x \in \range(\textbf{r})$ by checking if for all $i$, $x_i \in a_i + k\cdot c_i$ for some positive integer $k$ and $x_i \leq b_i$, which can be accomplished in $O(d)$ time. 
\end{enumerate}
\end{proof}

We can now invoke Theorem~\ref{thm:unknownM} to obtain the following result.

\kleearith*

\subsection{Test Coverage Estimation}\label{sec:coverage}
We now focus on the test coverage estimation problem and provide the proof of
Corollary~\ref{corr:cov}. Through out this section, we assume that all computations use random access to the bits of the items $\mathbf{a}$ in the stream. 
We begin with the following lemma.

\begin{lemma}
For each $\mathbf{a}_i \in \{0,1\}^n$, the set $\Cov_t(\mathbf{a}_i)$ belongs to
the Delphic family, assuming there is random access to the input.
\end{lemma}

\begin{proof}
Recall that
$
\Cov_t(\mathbf{a}_i)
= \{(T,\mathbf{y}) \mid T \subseteq [n], |T|=t, \mathbf{y} \in \{0,1\}^t,
\text{ and the }\text{ restriction of } \mathbf{a}_i \text{ to } T \text{ equals }
\mathbf{y}\}.
$
\begin{enumerate}
    \item The size of $\Cov_t(\mathbf{a}_i)$ is exactly $\binom{n}{t}$, since for
    every $T \subseteq [n]$ with $|T|=t$ there is a unique string
    $\mathbf{y} \in \{0,1\}^t$ induced by $\mathbf{a}_i$. This value can be
    computed in $O(t\log n)$ time.
    \item To draw a uniform random sample from $\Cov_t(\mathbf{a}_i)$, we sample
    a subset $T \subseteq [n]$ of size $t$ uniformly at random, and then return
    the pair $(T,\mathbf{y})$, where $\mathbf{y}$ is the restriction of
$\mathbf{a}_i$ to the coordinates in $T$. This can be implemented in
    $O(t \log n)$ time, since there is random access to the input.
    \item Given a pair $(T,\mathbf{y})$, we can test whether
    $(T,\mathbf{y}) \in \Cov_t(\mathbf{a}_i)$ by computing the restriction of
    $\mathbf{a}_i$ to $T$ and comparing it with $\mathbf{y}$, which takes
    $O(t \log n)$ time.
\end{enumerate}
All three operations are therefore supported efficiently, and
$\Cov_t(\mathbf{a}_i)$ is Delphic.
\end{proof}

Since each $\mathbf{a}_i$ implicitly represents the set
$\Cov_t(\mathbf{a}_i)$, Corollary~\ref{corr:cov} follows immediately from
Theorem~\ref{thm:unknownM}.

\covercorr*

We next investigate whether the space complexity in
Corollary~\ref{corr:cov} is optimal. The following observation shows that the
space complexity can be improved, albeit at the cost of significantly larger
update time.

\begin{observation}\label{obs:bestspace}
There is a streaming algorithm that, given a stream
$\mathcal{A} = \langle \mathbf{a}_1, \ldots, \mathbf{a}_M \rangle$, computes an
$(\varepsilon,\delta)$-approximation of  $|\Cov_t(\mathcal{A})|$ using
$O\!\left((t \log n + \varepsilon^{-2}) \log \frac{1}{\delta}\right)$ space and
$O(n^t)$ update time.
\end{observation}

\begin{proof}
For each incoming item $\mathbf{a}_i$, we explicitly enumerate all
$\binom{n}{t}$ pairs $(T,\mathbf{y}) \in \Cov_t(\mathbf{a}_i)$ and feed them into
a data stream. We then apply the $\mathbb{F}_0$-estimation algorithm of Kane,
Nelson, and Woodruff~\cite{KNW10} to estimate the number of distinct elements in
the resulting stream.

The reduction requires $O(t \log n)$ space to represent each element
$(T,\mathbf{y})$ and $O(n^t)$ time per update due to explicit enumeration. The
$\mathbb{F}_0$ algorithm itself uses
$O\!\left((t \log n + \varepsilon^{-2}) \log \frac{1}{\delta}\right)$ space and
constant update time per element, yielding the stated bounds.
\end{proof}

This observation raises the natural question of whether one can achieve a more
favorable trade-off between space and update time. In particular, we seek an
algorithm whose space complexity improves upon
Corollary~\ref{corr:cov} while avoiding the exponential update time of
Observation~\ref{obs:bestspace}. To this end, we next present a hashing-based
approach that replaces the $O(n^t)$ update time with polynomially many calls to
an NP oracle.

\paragraph{Alternate view of $\Cov_t(\mathbf{a})$.}
Recall, each element of $\Cov_t(\mathbf{a})$ is naturally represented as a pair
$(T,\mathbf{y})$, where $T$ is an ordered $t$-subset of $[n]$ and
$\mathbf{y}\in\{0,1\}^t$. We encode such a pair using a binary string of length
$O(t\log n)$ by representing the elements of $T$ in sorted order together with
the string $\mathbf{y}$. Accordingly, we view the universe of coverage elements
as a subset of $\{0,1\}^{O(t\log n)}$.

\paragraph{Overview of the algorithm.}
Algorithm~\ref{alg:hashing} can be viewed as an adaptation of the classical
$F_0$-estimation algorithm of Gibbons and Tirthapura~\cite{GT01}. In our setting,
each stream element $\mathbf{a}$ implicitly represents the associated set
$\Cov_t(\mathbf{a})$. The algorithm selects
$\mathsf{R} = O(\log(1/\delta))$ pairwise independent hash functions
$H[1],\ldots,H[\mathsf{R}]$.

For each hash function $H[i]$, the algorithm maintains:
(i) an integer level parameter $m[i] \ge 0$, and
(ii) a set $\mathcal{X}[i]$ containing elements
$z \in \bigcup_j \Cov_t(\mathbf{a}_j)$ such that
$H[i]_{m[i]}(z) = 0^{m[i]}$.
To ensure bounded space usage, the algorithm enforces the invariant
$|\mathcal{X}[i]| < \thresh$ for a suitable threshold $\thresh$.

Upon receiving a new stream element $\mathbf{a}$, the algorithm attempts to add
the elements of
$H[i]_{m[i]}^{-1}(0^{m[i]}) \cap \Cov_t(\mathbf{a})$
to $\mathcal{X}[i]$. If adding these elements would violate the invariant, the
algorithm increments $m[i]$ and refines $\mathcal{X}[i]$ by intersecting it with
$H[i]_{m[i]}^{-1}(0^{m[i]})$. This process repeats until the invariant is
restored.

\algrenewcommand\algorithmicindent{1em}

\begin{algorithm}
\caption{$\hashing$}\label{alg:hashing}
\begin{algorithmic}[1]
\State $\thresh \gets 1 + 9.84(1 + 1/\varepsilon^2)$
\State $\mathsf{R} \gets 35\log(1/\delta)$
\State $m[1:\mathsf{R}] \gets 0$; $\mathcal{X}[1:\mathsf{R}] \gets \emptyset$
\State $H \gets \mathrm{ChooseHashFunctions}(\mathsf{R})$
\While{\textbf{not} end of stream}
    \State $\mathbf{a} \gets \textsc{input}()$
    \For{$i = 1$ to $\mathsf{R}$}
        \While{\textbf{true}}
            \If{$\bigl|
            (H[i]_{m[i]}^{-1}(0^{m[i]}) \cap \Cov_t(\mathbf{a})) \cup \mathcal{X}[i]
            \bigr| 
            < \thresh $}\label{line:hashing-check}
                \State $\mathcal{X}[i] \gets
                \mathcal{X}[i] \cup
                (H[i]_{m[i]}^{-1}(0^{m[i]}) \cap \Cov_t(\mathbf{a}))$
                \State \textbf{break}
            \Else
                \State $m[i] \gets m[i] + 1$\label{line:hashing-fix-begin}
                \State $\mathcal{X}[i] \gets
                \mathcal{X}[i] \cap H[i]_{m[i]}^{-1}(0^{m[i]})$
                \label{line:hashing-fix-refine}
            \EndIf
        \EndWhile
    \EndFor
\EndWhile
\State \Return
$\mathrm{Median}\bigl(\{|\mathcal{X}[i]| \cdot 2^{m[i]}\}_{i=1}^{\mathsf{R}}\bigr)$
\end{algorithmic}
\end{algorithm}

The following lemma proves correctness of {\hashing}. 
\begin{restatable}{lemma}{hashingcorrect}\label{lem:hashing-correctness}
Let $
c := \mathrm{Median}\bigl(\{|\mathcal{X}[i]| \cdot 2^{m[i]}\}_{i=1}^{\mathsf{R}}\bigr).$
Then
\begin{align*}
\Pr\left[
\frac{|\Cov_t(\mathcal{A})|}{1+\varepsilon}
\le c \le
(1+\varepsilon)|\Cov_t(\mathcal{A})|
\right] \ge 1-\delta.
\end{align*}
\end{restatable}

While allowing larger constants in the definition of $\thresh$ would permit a
direct application of the analysis of Gibbons and Tirthapura~\cite{GT01}, we
instead provide an alternate proof inspired by techniques of Chakraborty, Meel,
and Vardi~\cite{CMV16} and Meel--Akshay~\cite{MA20}. The proof exploits the nested
structure of the sets $H[i]_{m[i]}^{-1}(0^{m[i]})$, which has been observed to
yield improved constants in hashing-based approximate counting. The proof is
deferred to the appendix.

\paragraph{Time complexity.}
The dominant cost of the algorithm arises from evaluating the condition in
line~\ref{line:hashing-check}. Membership testing
$z \in \Cov_t(\mathbf{a})$ can be performed in $O(t\log n)$ time using the
representation described above.

The following lemma, adapted from Lemma~3.7 of Bellare, Goldreich, and
Petrank~\cite{Bellare00}, characterizes the complexity of this check.

\begin{lemma}\label{lem:hashing-complexity}
Given a hash function $h \in \Hteop(t\log n,m)$ and a set $\mathcal{X}$, there exist
languages $M_1,M_2 \in \mathrm{NP}$ and a polynomial-time oracle algorithm
$\mathcal{A}^{M_1,M_2}$ such that, on input $(h,\mathbf{a},m,\thresh,\mathcal{X})$,
the algorithm outputs $0$ if
$\bigl|(h^{-1}(0^m) \cap \Cov_t(\mathbf{a})) \cup \mathcal{X}\bigr| \ge \thresh$,
and otherwise enumerates the elements of
$h^{-1}(0^m) \cap \Cov_t(\mathbf{a})$.
The algorithm makes $O(\thresh \cdot t \cdot \log n)$ oracle calls and uses
$O(\thresh\cdot t\cdot \log n)$ space.
\end{lemma}

The proof follows the structure of~\cite[Lemma~3.7]{Bellare00} and is omitted.

Combining Lemmas~\ref{lem:hashing-correctness} and
\ref{lem:hashing-complexity} establishes Theorem~\ref{thm:smallspace}.

\paragraph{Practical considerations.}
Although the algorithm relies on NP-oracle access, this does not preclude
practical applicability. Modern SAT solvers routinely handle instances with
millions of variables~\cite{MLM09}. Moreover, the oracle queries in
Algorithm~\ref{alg:hashing} can be expressed as CNF-XOR formulas, which have
received significant attention in the context of hashing-based model
counting~\cite{Stockmeyer83,CMV13b,CMV16,SM19,SGM20}. We leave the development of
practical implementations to future work.

\subsection{Model Counting of DNF}\label{sec:dnf}
Consider a DNF formula $\varphi = D_1 \vee D_2 \vee \ldots \vee D_M$, wherein each
$D_i$ is a term defined over a set of $n$ Boolean variables. We denote the set of
all satisfying assignments of $\varphi$ by $\satisfying{\varphi}$. Given
$\varphi$, the problem of model counting seeks to compute
$|\satisfying{\varphi}|$. We focus on the problem of model counting for DNF
formulas in the streaming setting, i.e., where the terms $D_i$ arrive one by one
over a stream. We first begin with the following lemma.

\begin{lemma}
For each $D_i$, $\satisfying{D_i}$ belongs to the Delphic family.
\end{lemma}
\begin{proof}
Note that $\satisfying{D_i} \subseteq \{0,1\}^n$. Let $|D_i|$ denote the number of
literals in the term $D_i$.
\begin{itemize}
    \item $|\satisfying{D_i}| = 2^{n-|D_i|}$, which can be computed in $O(n)$ time.
    \item Drawing a uniform sample from $\satisfying{D_i}$ amounts to drawing
    $n-|D_i|$ bits uniformly at random, which can be accomplished in $O(n)$ time.
    \item Finally, checking whether $x \in \satisfying{D_i}$ can be accomplished
    in $O(n)$ time.
\end{itemize}
\end{proof}

Since each $D_i$ implicitly represents $\satisfying{D_i}$, the following
corollary follows from Theorem~\ref{thm:unknownM}.

\dnfcorr*

\section{Conclusion and Future Outlook}\label{sec:conclusion}

To summarize, our investigations led us to design a surprisingly simple yet efficient scheme for computing the size of the union of sets belonging to a Delphic family. We then show that the notion of Delphic sets can capture three fundamental problems in the streaming setting: Klee's measure problem, coverage estimation problem, and DNF counting. For each of these problems, we provide efficient streaming algorithms.

Crucially, we believe the simplicity of our scheme should make it amenable to practical implementation and adoption. From a technical perspective, we sketch out three directions of interest:

\begin{description}

\item[Generalization of Delphic Sets] In this work, we limited our focus to three problems to showcase the generalizability of the notion of Delphic sets. As a future work, an interesting direction of work would be to study other streaming problems that reduce to Delphic sets. To this end, one line of work would be to relax the requirement of $O(\log |\Omega|)$ time to $O(\log |\Omega|)^{O(1)}$ for membership, counting, and sampling queries. 

\item[Higher Moments] the computation of union of sets corresponds to $\mathbb{F}_0$ (0-th frequency moment) estimation. In this context, a natural question would be whether we can generalize our sampling-based strategy for $\mathbb{F}_k$, i.e., k-th frequency moment estimation. 

\item[Beyond Insertion-Only Streams] The framework presented in this paper handles insertion-only streams. The past two decades have witnessed a long line of work on richer turnstile models that allow deletion. Therefore, an interesting direction of future work would be to explore sampling-based framework for turnstile model. 

\item [Complexity independent of M] The space and time complexity of {\hybrid} has logarithmic dependence on $M$, which is in line with the $O^{*}(1)$ notation introduced by Tirthapura and Woodruff ~\cite{TW12}. However, observe that we can cast the distinct element problem as a special case of union of Delphic sets wherein every set is simply a singleton. In case of distinct element problem, the algorithms without dependence on $M$ are known. Therefore, an interesting direction for future work would be to investigate whether sampling-based framework can lead to algorithms whose space and update time complexity are independent of $M$. 

\end{description}

\begin{acks}
The revised  (and arguably much cleaner) proof of {$\hybrid$} is based on the proof sketch due to Divesh Aggarwal and Maciej Obremski, which is published in a follow-up work on DNF counting~\cite{SACMO23}.
We thank the anonymous reviewers of PODS 21 for valuable comments. Meel was supported in part by National Research Foundation Singapore under its NRF Fellowship Programme[NRF-NRFFAI1-2019-0004 ] and AI Singapore Programme [AISG-RP-2018-005], and NUS ODPRT Grant [R-252-000-685-13].
Vinodchandran was supported in part by NSF CCF-184908 and NSF HDR:TRIPODS-1934884 awards.
\end{acks}

\clearpage

\newpage

\section*{Appendix}

We provide a Proof of Lemma~\ref{lem:hashing-correctness}. We first restate the lemma below.

\hashingcorrect*

\begin{proof}

Let
\begin{align*}
N &:= |\Cov_t(\mathcal{A})|, \qquad
\alpha := \frac{\varepsilon}{1+\varepsilon}.
\end{align*}
Fix an arbitrary ordering $\Cov_t(\mathcal{A})=\{z_1,\ldots,z_N\}$ and fix any
$i\in[\mathsf{R}]$.  At the end of the stream, Algorithm~\ref{alg:hashing}
maintains
\begin{align*}
\mathcal{X}[i]
&= \Cov_t(\mathcal{A}) \cap H[i]_{m[i]}^{-1}(0^{m[i]})
\end{align*}
Define the target interval
\begin{align*}
I_\good
&:= \left[\frac{N}{1+\varepsilon},\ (1+\varepsilon)N\right],
\end{align*}
and let $\bad_i$ be the event that $|\mathcal{X}[i]|\cdot 2^{m[i]}\notin I_\good$.

Let $\ell$ be the output length of the hash functions. For each
$j\in\{0,1,\ldots,\ell\}$, define
\begin{align*}
\Cnt{i}{j}
&:= \bigl|\Cov_t(\mathcal{A}) \cap H[i]_{j}^{-1}(0^{j})\bigr|,
\qquad
\mu_j := \frac{N}{2^j},
\end{align*}
and define events
\begin{align*}
B_{i,j}
&:= \bigl(\Cnt{i}{j}\le \thresh-1\bigr),
\\
L_{i,j}
&:= \left(\Cnt{i}{j}<\frac{\mu_j}{1+\varepsilon}\right),
\\
U_{i,j}
&:= \bigl(\Cnt{i}{j}>(1+\alpha)\mu_j\bigr).
\end{align*}
Also define
\begin{align*}
E_{i,j}
&:= \overline{B_{i,j-1}} \cap B_{i,j} \cap (L_{i,j}\cup U_{i,j}),
\end{align*}
with the convention that $B_{i,-1}$ is false (so $\overline{B_{i,-1}}$ is true).

For $\bad_i$ to occur, the algorithm must terminate at some level $j$ (captured
by $\overline{B_{i,j-1}}\cap B_{i,j}$) and the scaled count must fall outside
$I_\good$ (captured by $L_{i,j}\cup U_{i,j}$). Hence
\begin{align}
\label{eq:bad-ub}
\Pr[\bad_i]
\le
\Pr\!\left[\bigcup_{j=0}^{\ell} E_{i,j}\right].
\end{align}
We only obtain an upper bound (and not equality) because $I_\good$ has upper
endpoint $(1+\varepsilon)N$ whereas $U_{i,j}$ uses the factor
$1+\alpha \le 1+\varepsilon$.

Let $m^*$ be the smallest integer $j$ such that
$(1+\varepsilon)\mu_j \le \thresh-1.$
We will use the monotonicity (nesting) of the prefix slices:
\begin{align}
\label{eq:thresh_monotone}
B_{i,j} \Longrightarrow B_{i,j+1}
\qquad\text{for all } j\in\{0,\ldots,\ell-1\}.
\end{align}

\begin{claim}\label{lm:upperbound}
\begin{align*}
\Pr[\bad_i]
\le\ &\Pr[B_{i,m^*-3}] + \Pr[L_{i,m^*-2}]
\\
&\qquad + \Pr[L_{i,m^*-1}] + \Pr[L_{i,m^*}\cup U_{i,m^*}].
\end{align*}
\end{claim}

\begin{proof}
We make three observations.

\noindent\textbf{O1.} For all $j\le m^*-3$, we have $\mu_j/(1+\varepsilon)\ge \thresh$.
Thus $B_{i,j}\cap U_{i,j}=\emptyset$ and $B_{i,j}\cap L_{i,j}=B_{i,j}$, so
\begin{align*}
\bigcup_{j=0}^{m^*-3} E_{i,j}
&\subseteq \bigcup_{j=0}^{m^*-3} \bigl(\overline{B_{i,j-1}}\cap B_{i,j}\bigr)
\\
&\subseteq \bigcup_{j=0}^{m^*-3} B_{i,j}
\subseteq B_{i,m^*-3},
\end{align*}
where the last containment uses~\eqref{eq:thresh_monotone}. Hence $
\Pr\!\left[\bigcup_{j=0}^{m^*-3} E_{i,j}\right]
\le \Pr[B_{i,m^*-3}]$

\noindent\textbf{O2.} For $j\in\{m^*-2,m^*-1\}$, we have $\thresh \le (1+\alpha)\mu_j$,
so $B_{i,j}\cap U_{i,j}=\emptyset$. Since $B_{i,j}\cap L_{i,j}\subseteq L_{i,j}$,
\begin{align*}
\bigcup_{j=m^*-2}^{m^*-1} E_{i,j}
&\subseteq L_{i,m^*-2}\cup L_{i,m^*-1},
\\
\Pr\!\left[\bigcup_{j=m^*-2}^{m^*-1} E_{i,j}\right]
&\le \Pr[L_{i,m^*-2}] + \Pr[L_{i,m^*-1}].
\end{align*}

\noindent\textbf{O3.} For all $j\ge m^*$, we have $\thresh \ge (1+\alpha)\mu_j$,
which implies $\overline{B_{i,j}}\subseteq U_{i,j}$. Using again that
$\overline{B_{i,j}}\subseteq \overline{B_{i,j-1}}$ by~\eqref{eq:thresh_monotone},
we obtain the containment
\begin{align*}
\bigcup_{j=m^*}^{\ell} E_{i,j}
&\subseteq \overline{B_{i,m^*}}
 \cup
\bigl(\overline{B_{i,m^*-1}} \cap B_{i,m^*} \cap (L_{i,m^*}\cup U_{i,m^*})\bigr)
\\
&\subseteq \overline{B_{i,m^*}} \cup L_{i,m^*} \cup U_{i,m^*}
\subseteq L_{i,m^*}\cup U_{i,m^*},
\end{align*}
and therefore
\begin{align*}
\Pr\!\left[\bigcup_{j=m^*}^{\ell} E_{i,j}\right]
\le \Pr[L_{i,m^*}\cup U_{i,m^*}].
\end{align*}

Combining O1--O3 with~\eqref{eq:bad-ub} proves the claim.
\end{proof}

\begin{claim}\label{lm:aux-bounds}
The following bounds hold:
\begin{align*}
\Pr[L_{i,m^*}\cup U_{i,m^*}] &\le \frac{1}{4.92}, \qquad &
\Pr[L_{i,m^*-1}] &\le \frac{1}{10.84},\\
\Pr[L_{i,m^*-2}] &\le \frac{1}{20.68}  & 
\Pr[B_{i,m^*-3}] &\le \frac{1}{62.5}.
\end{align*}
\end{claim}

\begin{proof}
We have
\begin{align*}
\Pr[B_{i,j}] &= \Pr\!\left[\Cnt{i}{j}\le \thresh-1\right],\\
\Pr[L_{i,j}] &= \Pr\!\left[\Cnt{i}{j}\le \frac{\mu_j}{1+\varepsilon}\right],\\
\Pr[L_{i,j}\cup U_{i,j}]
&=
\Pr\!\left[\bigl|\Cnt{i}{j}-\mu_j\bigr|\ge \alpha\,\mu_j\right].
\end{align*}
Substituting the values of $m^*$, $\thresh$, and $\mu_j$, and applying
Chebyshev's and Paley--Zygmund inequalities yields the stated bounds.
\end{proof}

By Claim~\ref{lm:upperbound} and Claim~\ref{lm:aux-bounds}, we obtain
$\Pr[\bad_i]\le p$ for a constant $p<1/2$.  Define
\begin{align*}
c_i &:= |\mathcal{X}[i]|\cdot 2^{m[i]}, \qquad
c := \mathrm{Median}\bigl(\{c_i\}_{i=1}^{\mathsf{R}}\bigr).
\end{align*}
Since the hash functions $H[1],\ldots,H[\mathsf{R}]$ are chosen independently,
the events $\bad_1,\ldots,\bad_{\mathsf{R}}$ are independent. A Chernoff bound
then implies that for $\mathsf{R}=35\log(1/\delta)$, the median $c$ lies in
$I_\good$ with probability at least $1-\delta$, completing the proof.
\end{proof}

\end{document}